\documentclass[runningheads,a4paper,envcountsame]{llncs-3icmcta}
\usepackage{amssymb}
\newcommand{\keywords}[1]{\par\addvspace\baselineskip
\noindent\keywordname\enspace\ignorespaces#1}
\begin{document}
\title{A generalisation of the Gilbert-Varshamov bound and its
asymptotic evaluation}
\titlerunning{A generalisation of the GV bound and its asymptotic evaluation}
\author{Ludo Tolhuizen} \email{}\institute{Philips Research Europe, Eindhoven, The
Netherlands} \maketitle

\begin{abstract}
The Gilbert-Varshamov (GV) lower bound on the maximum cardinality of
a $q$-ary code of length $n$ with minimum Hamming distance at least
$d$ can be obtained by application of Tur\'{a}n's theorem to the
graph with vertex set $\{0,1,\ldots ,q-1\}^n$ in which two vertices
are joined if and only if their Hamming distance is at least $d$. We
generalize the GV bound by applying Tur\'{a}n's theorem to the graph
with vertex set $C^n$, where $C$ is a $q$-ary code of length $m$ and
two vertices are joined if and only if their Hamming distance at
least $d$. We asymptotically evaluate the resulting bound for
$n\rightarrow \infty$ and $d\sim \delta mn$ for fixed $\delta >0$,
and derive conditions on the distance distribution of $C$ that are
necessary and sufficient for the asymptotic generalized bound to
beat the asymptotic GV bound. By invoking the Delsarte inequalities,
we conclude that no improvement on the asymptotic GV bound is
obtained.
\\ By using a sharpening of
Tur\'{a}n's theorem due to Caro and Wei, we improve on our bound. It
is undecided if there exists a code $C$ for which the improved bound
can beat the asymptotic GV bound.
\end{abstract}
\section{Introduction}
Let $A_q(n,d)$ denote the maximum cardinality of code of length $n$
and minimum Hamming distance at least $d$ over an alphabet $Q$ with
$q$ letters. Moreover, for $0\leq\delta\leq 1$, let
$\alpha_q(\delta)$ denote the limes superior of the maximum rate of
$q$-ary codes of relative distance $\delta$, that is,
\[ \alpha_q(\delta) = \limsup_{n\rightarrow\infty} \frac{1}{n}
\log_q A_q(n,\delta n) .\] According to the asymptotic Plotkin bound
\cite[Thm.\ 5.2.5]{10:vL} $\alpha_q(\delta)=0$ for $\delta\geq
1-\frac{1}{q}$; for $0<\delta < 1-\frac{1}{q}$, the value of
$\alpha_q(\delta)$ is unknown.

The Gilbert-Varshamov (GV) bound  \cite[Thm.\ 5.1.7]{10:vL} states
that
 \[ A_q(n,d) \geq q^n \left/ \sum_{i=0}^{d-1} {n\choose i}(q-1)^i.
 \right. \]
The asymptotic version of the GV bound \cite[Thm.\ 5.1.9]{10:vL}
reads as follows:
\begin{equation}\label{10:eq:gvas}
\mbox{ for } 0 \leq \delta \leq 1-\frac{1}{q}, \mbox{ we have that }
\alpha_q(\delta) \geq 1-h_q(\delta),
\end{equation}
 where $h_q$ is the $q$-ary entropy function, defined as
  \[ h_q(x) = -x\log_q(x) - (1-x)\log_q(1-x) + x\log_q(q-1) . \]
To the best of the author's knowledge, no lower bound on
$\alpha_q(\delta)$ improving on (\ref{10:eq:gvas}) is known for $q <
46$. For an extensive survey on literature on the Gilbert-Varshamov
 bound and improvements on it, we refer to \cite{10:Jiang}.

 In \cite{10:Tol}, it was observed that the GV bound can be obtained by
 application of Tur\'{a}n's theorem \cite[Thm.\ 3.2.1]{10:AlonSpen}
 to the graph with vertex set $Q^n$, in which two vertices are joined by
 an edge if and only if their Hamming distance is at least $d$.
 By using this graph-theoretical approach and applying a refined version of
 Tur\'{a}n's theorem for locally sparse graphs, Jiang and Vardy \cite{10:Jiang}
 obtained an improvement of the GV bound for binary codes with a multiplicative factor
 $n$.
 %In \cite{10:Gab}, Gaborit and Z\'{e}mor show that certain families of
 %linear binary $[n,n/2]$ random double circulant codes satisfy this
% improved GV bound.
This result was generalized to $q$-ary codes by Vu and Wu
%\cite{10:Vu}.
 who proved the following \cite[Thm.\ 1.2]{10:Vu}.
 \begin{theorem}\label{10:th:Vu}
 Let q be a fixed positive integer and let $\beta, \beta^{\prime}$
 be constants satisfying $0< \beta^{\prime} < \beta <
 \frac{q-1}{q}$. There is a positive constant $c$ depending on $q$
 and $\beta$ such that for any $\beta^{\prime}n < d < \beta n$,
 \[ A_q(n,d+1) \geq c \frac{q^n}{\sum_{i=0}^{d} {n\choose i}(q-1)^i}
 n . \]
 \end{theorem}
%In \cite{10:Gab}, Gaborit and Z\'{e}mor show that certain families
%of linear binary $[n,n/2]$ random double circulant codes satisfy the
%Jiang-Vardy bound.  Obviously, the Jiang-Vardy bound and Vu-Wu bound
%do not yield an improvement to the asymptotic GV bound
%(\ref{10:eq:gvas}).

 In this paper, we use the graph-theoretical approach to obtain a generalization of the GV
 bound. Again, we consider a graph in which two vertices are
 joined by an edge if and only if their Hamming distance is at least
 $d$. The vertex set does not equal $Q^n$, but it equals
 $C^n$, where $C\subseteq Q^m$ is a fixed $q$-ary code of length
 $m$. We use Tur\'{a}n's theorem to obtain a lower bound on the size
 of the largest clique in this graph, and employ a bounding technique from
 \cite{10:Bur} to obtain a manageable asymptotic expression. We analyze
 the generalized asymptotic GV bound,  and by employing the Delsarte
 inequalities \cite[Sec.\ 5.3]{10:vL},
 we infer that it cannot improve the
 asymptotic GV bound. We end with an improvement of our bound based
 on a sharpening of Tur\'{a}n's theorem due to Caro and Wei, and derive a necessary and sufficient condition
 on $C$ for this improved bound to beat the asymptotic GV bound. We have not been able to
 decide if there exist codes $C$  satisfying this condition.

 %Throughout the paper, we use the following notation.
% We consider codes over the alphabet $Q$ with $q=|Q|$ letters.
% The Hamming distance between two vectors ${\bf x}$ and ${\bf y}$ of
%equal length is denoted by $d({\bf x},{\bf y})$. The largest size of
%a code $q$-ary code of length $n$ and minimum Hamming distance at
%least $d$ is denoted by $A_q(n,d)$, and for $0\leq \delta \leq 1$,
%we define
%\[ \alpha_q(\delta) = \liminf \frac{1}{n} \log_q A(n,\lfloor \delta
%n\rfloor ) . \]
%
% \[ \alpha_q(\delta) \geq 1-h_q(\delta), \]
% where $h_q$ is the $q$-ary entropy code, defined as
%  \[ h_q(x) = -x\log_q(x) - (1-x)\log_q(1-x) + x\log_q(q-1) . \]
\section{Prerequisites}
\subsection{Tur\'{a}n's theorem}
 Let $G$ be a simple graph without loops. A {\em
clique} in $G$ is a set of vertices of which each pair is joined by
an edge. It is intuitively clear that a graph with many edges should
contain a large clique. This is quantified by Tur\'{a}n's theorem,
of which we use the following version (for a proof, see e.g.\
\cite[Thm. 3.2.1]{10:AlonSpen}).
%\begin{theorem}\label{th:Turan}
%Let $G=(V,E)$ be a simple graph without loops with $v$ vertices and
%$e$ edges. Write $v=(k-1)t+r$ with $1\leq r\leq k-1$. If
% \[ e > \frac{k-2}{2(k-1)}v^2 - \frac{r(k-r-1)}{2(k-1)}, \] then $G$
% contains a clique of size $k$.
%\end{theorem}
%Clearly, Tur\'{a}n's theorem implies that a simple graph with $v$
%vertices and $e$ edges contains a clique of size $k$ if
% \[ e > \frac{k-2}{2(k-1)}v^2,  \] or stated differently:
% a simple graph with $v$ vertices and $e$ edges contains a clique of
% size $\lceil \frac{v^2}{v^2-2e}\rceil $.
\begin{theorem}\label{10:th:Turan}
A simple graph without loops with $v$ vertices and $e$ edges
contains a clique of size at least $v^2 / \left( v^2-2e \right)$.
\end{theorem}
\subsection{Distance enumerator of a code}
Let $C\subseteq Q^m$. For 0$\leq j\leq m$, we define $B_j$ as
\[ B_j = \frac{1}{|C|} | \{ (x,y)\in C^2 \mid d(x,y)=j \}| . \]
We define the distance enumerator polynomial $B(x)$ of $C$ as $B(x)
= \sum_{j=0}^m B_j x^j$. It is easy to check that for each $n\geq
1$, the code $C^n$ has distance  enumerator polynomial $\left(
B(x)\right)^n$.
%Now, let $C$ and $D$ be codes of length $n$ and $m$, respectively.
%The direct product  $C\times D$ of $C$ and $D$ is the code of length
%$n+m$ defined as
%\[ C\times D = \{ ({\bf c},{\bf d}) \mid {\bf c}\in C, {\bf d} \in
%D\} . \] It is well-known and easy to check that the distance
%enumerator polynomials $B_C(x),B_D(x)$ and $B_{C\times D}(x)$ of
%$C$, $D$ and $C\times D$, respectively, are related by the following
%equation
% \[ B_{C\times D}(x) = B_C(x) B_D(x) . \]
% As a result, one readily finds that for each integer $k\geq
% 0$, the weight enumerator of $C^k = C\times C\times \cdots C$ ($k$
% times) is given by

 %\[ B_{C^k} (x) =  (B(x))^k . \]

%By applying  Tur\'{a}n's bound to $G$, we find that there exists a
%subcode of $C$ with minimum distance at least $d$ and cardinality at
%least
%\[ |C| / \sum_{j=1}^{d-1} B_j . \]
\subsection{A bounding lemma}
The following lemma is similar to a bounding technique that can be
found in \cite{10:Bur}.
\begin{lemma}\label{10:lem:Bur} Let $f(x)=\sum_{j=0}^n f_j x^j$ be a polynomial with non-negative
coefficients, i.e., $f_j\geq 0$ for all $j$. For each $k, 0\leq
k\leq n$ and each $x\in (0,1]$, we have that \[ \sum_{j=0}^k f_j
\leq \frac{f(x)}{x^k} . \]
\end{lemma}
\begin{proof}
Let $0\leq x\leq 1$.  Whenever $0\leq j\leq k$, we have that
$x^j\geq x^k$, and so $f(x)=\sum_{j=0}^n f_j x^j \geq \sum_{j=0}^k
f_j x^ j \geq \sum_{j=0}^k f_j x^k$.
 $\qed$ \end{proof}
%  For fixed $\delta\in
%(0,1)$, we have that
%\[ \sum_{j=0}^{nm\delta} B_j^{(n)} \leq \frac{B(x)^n}{x^{\delta m
%n}} = \left[ \frac{B(x)}{x^{\delta m}} \right]^{n}, \] and so for
%each $x\in (0,1)$, we have that
%\[  \frac{1}{mn} \log \left(
%\sum_{j=0}^{nm\delta} B_j^{(n)}\right) \leq \frac{1}{m} B(x) -
%\delta\log x . \]
% As a consequence, there exists a sequence of codes with length $mn$
% with minimum distance $\delta mn$ and rate $R$ satisfying
% \[ R \geq \frac{1}{mn} ( \log |C|^n - \log \left(\sum_{j=0}^{mn\delta}
% B_j^ {(n)}\right) ) \geq \frac{1}{m}\log \left( \frac{|C|}{B(x)}\right) +
% \delta \log x . \]
% (NOTE: In fact, I think that using the methods from \cite{Bur}, it
% can be shown that
% \[ \lim_{n\rightarrow \infty} \frac{1}{mn} \log \left(
%\sum_{j=0}^{nm\delta} B_j^{(n)}\right) = \inf_{x\in (0,1)}
%\frac{1}{m} B(x) - \delta\log x . \]
\section{Main result and its proof}
\begin{theorem}\label{10:th:main}
Let $C\subseteq Q^m$ have distance enumerator polynomial $B(x)$. For
each $x\in (0,1]$, we have
\[ \alpha_q(\delta) \geq \frac{1}{m} \log_q \left(
\frac{|C|}{B(x)}\right)  + \delta \log_q x . \]
\end{theorem}
\begin{proof}
For each integer $n$, we consider the graph $G$ with vertex set
$C^n$ in which two vertices are joined by an edge if and only if
they have Hamming distance at least $d=\lceil \delta mn\rceil$. The
number of edges $e$ thus is given by
\[ e = |C|^n\frac{1}{2} \sum_{j=d}^m B_j^{(n)} , \]
where $B_j^{(n)}$ is the $j$-th coefficient of the distance
enumerator of $C^n$. Application of  Tur\'{a}n's bound to $G$ yields
the existence of a subcode $D$ of $C^n$ with minimum distance at
least $d$ and cardinality at least
\[ |C|^n / \sum_{j=0}^{d-1} B_j^{(n)} . \]
We now invoke Lemma~\ref{10:lem:Bur} and find that for each $x\in
(0,1]$,
\[ \sum_{j=0}^{d-1} B_j^{(n)} \leq \frac{B_{C^n}(x)}{x^{d-1}} =
\left( \frac{B(x)}{x^{{\delta}m}} \right)^n  . \] Combining the
above inequalities, we find that the code $D$ of length $mn$
satisfies
 \[ \frac{1}{mn}\log_q |D| \geq   \frac{1}{m}\log_q
 \left( \frac{|C|}{B(x)} \right) + \delta \log_q x. \]
 $\qed$
\end{proof}
It can be readily verified that Theorem~\ref{10:th:main} with $m=1$,
$C=Q$ and $x= \delta / (q-1)(1-\delta)$ yields the asymptotic GV
bound (\ref{10:eq:gvas}).
\\ \\
Theorem~\ref{10:th:main} contains a parameter $x$ that can be
optimized over. By straightforward differentiation, one finds that
the optimizing value $x$ satisfies the equation $x B^{\prime}(x) -
\delta m B(x) = 0$.  For a given code $C$, it seems in general a
hopeless task to obtain a closed expression for the largest bound on
$\alpha_q(\delta)$  that can be obtained from
Theorem~\ref{10:th:main}. We set ourselves instead a different goal,
{\em viz.} to determine if Theorem~\ref{10:th:main} can improve on
the asymptotic GV bound. To this end, for each $\delta\in
(0,1-\frac{1}{q})$ and $x\in (0,1)$, we define
\[ F(x,\delta) = \frac{1}{m} \log_q\left( \frac{|C|}{B(x)}\right) +
\delta\log_q x - (1-h_q(\delta)) . \] For a pair $(x,\delta)$
optimizing $F(x,\delta)$, we have that
\[ 0 = \partial F(x,\delta) / \partial \delta = \log_q x + h_q^{\prime}(\delta),
\mbox{ implying that } x = x_{\delta}:= \frac{\delta}{(q-1)(1-\delta)}. \]
Note that $x_{\delta} \leq 1$ whenever $\delta\leq 1-\frac{1}{q}$.
As $\delta\log_q x_{\delta} + h_q(\delta) = -\log_q(1-\delta)$, we
have that
\[ F(x_{\delta},\delta) =
  \frac{1}{m} \log_q (\frac{|C|}{B(x_{\delta})}) - 1 - \log_q(1-\delta) =
   -\frac{1}{m} \log_q\left(
   B(x_{\delta})
   (1-\delta)^m \frac{q^m}{|C|}\right) . \]
As a result, we have proved the following lemma.
\begin{lemma}\label{10:th:condition}
There exists a $\delta$ for which Theorem~\ref{10:th:main} yields an
improvement on the asymptotic GV bound if and only if for some
$\delta\in (0,1-\frac{1}{q})$
   \[
   \frac{q^m}{|C|}(1-\delta)^m B(\frac{\delta}{(q-1)(1-\delta)}) <
   1.
   \]
   \end{lemma}
\begin{theorem}\label{10:th:toobad}
The largest lower bound on $\alpha_q(\delta)$ that can be obtained
from \\ Theorem~\ref{10:th:main} is the asymptotic GV bound
$1-h_q(\delta)$.
\end{theorem}
\begin{proof}
By substituting $z=1-\frac{q}{q-1}\delta$ in the condition of
Lemma~\ref{10:th:condition},  we obtain the equivalent condition
that for some $z\in (0,1)$,
   \begin{equation}\label{10:condition2} \frac{1}{|C|} (1+(q-1)z)^m B(\frac{1-z}{1+(q-1)z})
   < 1. \end{equation}
We write the left hand side  of (\ref{10:condition2}) as the
polynomial $A(z)= \sum_{i=0}^m A_i z^i$.  By choosing $z=0$, we find
that $A_0=A(0)=\frac{1}{|C|}B(1)=1$. According to the Delsarte
inequalities, that form the basis of the linear programming bound
\cite[Sec.\ 5.3]{10:vL},  $A_i\geq 0$ for all $i$. So for $z\geq 0$,
we have that $A(z)\geq A_0=1. \;\;\;\;\qed$
\end{proof}
\section{Extension of the main result}
We extend our result by using the sharpening of Tur\'{a}n's theorem
from Theorem~\ref{10:th:Turansum} below; an elegant proof,
attributed to Caro and Wei, can be found in \cite[p.\
95]{10:AlonSpen}.
\begin{theorem}\label{10:th:Turansum}
Let $G=(V,E)$ be a simple graphs without loops. For each  $v\in V$,
let $d_v$ be the number of neighbours of $v$. Then $G$ contains a
clique of size at least
 \[ \sum_{v\in V} \frac{1}{|V|-d_v}. \]
 \end{theorem}

 By using a convexity argument and the fact that $\sum_{v\in V} d_v
 = 2|E|$, it can be shown that Theorem~\ref{10:th:Turansum} implies
 Theorem~\ref{10:th:Turan}, and that they give the same result for regular graphs,
 i.e.,
if all vertices have equally many neighbours. If we thus apply our
construction with a code $C$ for which the number of codewords at a
given distance from a word ${\bf c}\in C$ actually depends on the
choice of ${\bf c}$, we may improve our main result. For describing
this improvement, we introduce the following notion:  for a given
code $C$ and ${\bf c}\in C$, the local distance enumerator $B_{{\bf
c}}(x)$ is defined as
\[ B_{{\bf c}}(x) = \sum_{j=0}^m |\{ {\bf y}\in C\mid d({\bf c},{\bf
y}) = j\}| x^j. \]
%Note that the distance enumerator polynomial
%$B(x)$ is the average of all local distance enumerator polynomials
%$B_{\bf c}(x)$.
\begin{theorem}\label{10:th:LTgen} Let $C\subseteq Q^m$. For each $\delta\in (0,1)$ and $x\in (0,1)$, we have
\[ \alpha_q(\delta) \geq   \frac{1}{m}\log \left( \sum_{{\bf
c}\in C} \frac{1}{B_{\bf c}(x)} \right) + \delta \log_q x . \]
\end{theorem} Note that Theorem~\ref{10:th:LTgen} reduces to
Theorem~\ref{10:th:main} if all local distance distributions are
equal.
\begin{proof}
We apply Theorem~\ref{10:th:Turansum} to the graph with vertex set
$C^n$, in which two vertices are joined by an edge if and only if
they have Hamming distance at least $d$. In this way, we infer the
existence of a code of length $mn$ and minimum Hamming distance at
least $d$ of size at least
\[ \sum_{{\bf c}\in C^n} \frac{1}{|\{ {\bf y}\in C^n \mid
d({\bf c},{\bf y}) \leq d \}| }\] For each {\bf c}$\in C^n$ and
$x\in (0,1]$, we have that
\[ |\{y \in C^n \mid d({\bf c},{\bf y})\leq d\}| \leq \frac{B_{\bf c}^{(n)}(x)}{x^d}, \]
where $B_{\bf c}^{(n)}(x) = \sum_{j} |\{ {\bf y}\in C^n\mid d({\bf
c},{\bf y})=j\}| x^j. $ It is easy to see that for ${\bf c}=({\bf
c}_1,{\bf c}_2,\ldots ,{\bf c}_n)\in C^n$, we have that
 \[ B_{{\bf c}}^{(n)}(x) = \prod_{i=1}^n B_{{\bf c}_i}(x), \]
 and so there exists a code of length $mn$ and size at least
 \[ \sum_{{\bf c}={\bf c}_1,\ldots, {\bf c}_n \in C^n }\frac{x^d}{B_{{\bf
 c}}(x)} = x^d \sum_{({\bf c}_1,\ldots ,{\bf c}_n)\in C^n}
 \frac{1}{\prod_{i=1}^n B_{{\bf c}_i}(x)} = x^d \left( \sum_{{\bf
 c}\in C} \frac{1}{B_{\bf c}(x)} \right)^n, \]
 where the final equality can be proved by induction on $n.
 \;\;\;\Box$
 \end{proof}
\begin{lemma}\label{10:lem:LTgen}
Theorem~\ref{10:th:LTgen} improves on the asymptotic GV bound for
some $\delta \in (0,1-\frac{1}{q})$ if and only if for some $z\in
(0,1)$,
\[ \sum_{\bf c\in C} \frac{1}{(1+(q-1)z)^m B_{\bf c}\left(
\frac{1-z}{1+(q-1)z}\right)} > 1 . \]
\end{lemma}
\begin{proof}
Similar to the proof of Lemma~\ref{10:th:condition} and
Theorem~\ref{10:th:toobad}. $\;\;\;\;\qed$
\end{proof}
We have not been able to decide if there exist codes for which the
inequality in Lemma~\ref{10:lem:LTgen} is met. We note that for
certain codes $C$, e.g.\ for
 $C=\{0,1\}^3 \setminus \{(0,0,0)\}$, the left-hand side of the
 condition in Lemma~\ref{10:lem:LTgen} is monotonically decreasing in $z$,
 although not all individual terms have this property.
%In fact, for some small binary codes that we investigated, the
%left-hand side is monotonically decreasing on (0,1), whence its
%maximum is attained for $z=0$ and thus has value exactly 1. As an
%example, take $m=3$ and $C=\{ 0,1 \}^3\setminus \{(0,0,0)\}$. We
%then have that $B_{{\bf c}}(x) = (1-x)^3-x^{\mbox{\small{wt}}({\bf
%c})}$, and so
%\[ \sum_{{\bf c}\in C} \frac{1}{(1+z)^m B_c \left(
%\frac{1-z}{1+z}\right)} = \frac{3}{8-(1-z)(1+z)^2} +
%\frac{3}{8-(1-z)^2(1+z)} + \frac{1}{8-(1-z)^3}.\] It can readily be
%checked that the first term in the above sum is increasing on
%$(0,1/3)$ and decreasing on $(1/3,1)$, while the other two terms are
%decreasing in $z$. The sum of the three terms is monotonically
%decreasing on (0,1).

\end{document}